\documentclass[a4paper, reqno]{amsart}
\usepackage{a4wide}
\usepackage{amsmath,amsfonts,amssymb,amsthm,color}
\usepackage{comment}
\usepackage{enumitem}

\renewcommand\labelenumi{(\roman{enumi})}
\renewcommand\theenumi\labelenumi

\theoremstyle{plain}
\newtheorem{theorem}{Theorem}[section]

\newtheorem{corollary}[theorem]{Corollary}
\newtheorem{lemma}[theorem]{Lemma}

\theoremstyle{remark}
\newtheorem{remark}[theorem]{Remark}

\newcommand{\xx}{{\underline x}}

\newcommand{\fl}{\mathfrak{f}_\lambda}

\newcommand{\R}{{\mathbb R}}

\newcommand{\kk}{{\bf k}}

\newcommand{\x}{{\bf x}}

\newcommand{\yy}{{\underline y}}
\newcommand{\y}{{\bf y}}

\numberwithin{equation}{section}

\newcommand{\tx}{\widetilde{\bf x}}

\newcommand{\F}{\Pi}
\newcommand{\Fe}{\mathfrak{P}}
\newcommand{\Ids}{\mathcal{I}}
\newcommand{\Chern}{{\rm Ch}}
\newcommand{\ASl}{\mathcal{A}}

\newcommand{\N}{{\mathbb N}}
\newcommand{\Z}{{\mathbb Z}}

\newcommand{\Id}{\mathbf{1}}

\newcommand{\iu}{\mathrm{i}}

\newcommand{\set}[1]{\left\{ #1 \right\}}

\newcommand{\sub}[1]{_{\text{#1}}}

\DeclareMathOperator{\Tr}{Tr}

\newcommand{\BF}{\mathcal{U}_{BF}}
\newcommand{\BFZ}{\mathcal{U}_{BFZ}}

\title[Gap labelling for Bloch-Landau Hamiltonians]{Beyond Diophantine Wannier diagrams:\\ Gap labelling for Bloch-Landau Hamiltonians}
\author{Horia D. Cornean, Domenico Monaco \and Massimo Moscolari}
\date{April 10, 2020. arXiv version 2} 

\usepackage{hyperref}

\begin{document}


\begin{abstract}
It is well known that, given a $2d$ purely magnetic Landau Hamiltonian with a constant magnetic field $b$ which generates a magnetic flux $\varphi$ per unit area, then any spectral island $\sigma_b$ consisting of $M$ infinitely degenerate Landau levels carries an integrated density of states $\mathcal{I}_b=M \varphi$. Wannier later discovered a similar Diophantine relation expressing the integrated density of states of a gapped group of bands of the Hofstadter Hamiltonian as a linear function of the magnetic field flux with integer slope.

We extend this result to a gap labelling theorem for any $2d$ Bloch-Landau operator $H_b$ which also has a bounded $\Z^2$-periodic electric potential. Assume that $H_b$ has a spectral island $\sigma_b$ which remains isolated from the rest of the spectrum as long as $\varphi$ lies in a compact interval $[\varphi_1,\varphi_2]$. Then $\mathcal{I}_b=c_0+c_1\varphi$ on  such intervals, where the constant $c_0\in \mathbb{Q}$ while $c_1\in \Z$. The integer $c_1$ is the Chern marker of the spectral projection onto the spectral island $\sigma_b$. This result also implies that the Fermi projection on $\sigma_b$, albeit continuous in $b$ in the strong topology, is nowhere continuous in the norm topology if either $c_1\ne0$ or $c_1=0$ and $\varphi$ is rational.

Our proofs, otherwise elementary, do not use non-commutative geometry but are based on gauge covariant magnetic perturbation theory which we briefly review  for the sake of the reader. Moreover, our method allows us to extend the analysis to certain non-covariant systems having slowly varying magnetic fields.	
	
\bigskip


\medskip

\end{abstract}

\maketitle

\tableofcontents

\section{Introduction and main results}

It is by now textbook material \cite{Landau30,Pierls55} that each Landau level of an electron moving freely in 2-dimensions in the presence of a constant magnetic field $b$ carries a density of states per unit area equal to the magnetic field flux $\varphi=b/2\pi$, in a suitable system of physical units.

In 1978, Wannier \cite{Wannier78} realized by an ingenious counting argument that the integrated density of states of any isolated group of mini-bands of the Hofstadter Hamiltonian \cite{Hofstadter76}, a discrete analogue of the magnetic Laplacian, is a linear Diophantine function of the rational magnetic flux. Moreover, its slope is an integer which remains unchanged as long as the group of mini-bands under consideration remains isolated from the other ones. More specifically, consider a reference magnetic field $b_0$ such that $b_0 =2\pi p_0/q$, with $p_0,q \in \Z$ co-prime integers and $q$ a very large integer. If $\mathcal{I}_{b_0}$ denotes the integrated density of states associated to $M$ mini-bands of the Hofstadter Hamiltonian at magnetic field $b_0$, it holds that $\mathcal{I}_{b_0}= M/q$. Denoting by  $\mathcal{I}_{b}$ the integrated density of states associated to the same group of mini-bands (which now consists of a different number) but for $b= 2\pi p/q$, with integers $p,q$ and $p-p_0$ small enough, then
\begin{equation} \label{WannierDiagram}
q \, \mathcal{I}_{b}=M +c_1 (p-p_0), \quad c_1 \in \Z.
\end{equation}
Notice that the left-hand side of the above equality counts the number of charge carriers in a supercell of area $q$.
Without giving a formal proof, Wannier came to the natural conclusion that this relationship should also hold for all irrational values of the flux: this allowed him to label the gaps in the spectrum of the Hofstadter Hamiltonian by ``diagrams'' consisting of linear functions of magnetic flux with integer slopes. No wonder that his paper was rather cryptically entitled  {\it A Result Not Dependent on Rationality for Bloch Electrons in a Magnetic Field}. 

In 1982, starting from the linear response ansatz, St{\v r}eda \cite{Streda82a}  showed  that the Hall conductivity is proportional to the derivative with respect to the magnetic flux of the integrated density of states of the Fermi projection, provided the Fermi energy is in a gap. Then in \cite{Streda82} he used Wannier's result from 1978 in order to conclude that the Hall conductivity is proportional to an integer, namely $c_1$ in the above formula \eqref{WannierDiagram}. 

Still in 1982, Thouless {\sl et al.}\ \cite{TKNN82} showed that the Hall conductivity is proportional to the Chern number of the Fermi projection whenever the number of magnetic flux quanta per unit cell is rational, and thus identified the geometric origin underlying the integer $c_1$: this relation of the Hall conductivity with topological numbers was later clarified also by Avron {\sl et al.}, see \cite{AvronSeilerSimon83}. Reasoning in analogy with the Hofstadter model, and inspired by Wannier's work, Thouless and his collaborators concluded that the results should persist also at irrational values of the magnetic flux. This led Avron and Osadchy to produce ``colored Hofstadter butterflies'', where the gaps in the spectrum of the Hofstadter Hamiltonian are labelled according to their associated Chern number \cite{OsadchyAvron01,Avron04}.

For discrete and continuum gapped models of Bloch electrons, Wannier's result and its connection with the Chern marker (see \eqref{GenChern} below) for all real flux values were rigorously formulated by Bellissard \cite{Bellissard86,Bellissard89} in the language of non-commutative geometry. The equality argued by Wannier, generalizing \eqref{WannierDiagram} to any $b$, was dubbed ``gap labelling conjecture'' in \cite{Bellissard86}, and translated in a statement about the $K$-theory of certain crossed product $C^*$-algebras. Bellissard proved the gap labelling conjecture for aperiodic crystals without magnetic field in \cite{Bellissard92}, and the proof has been extended to more general quasi-crystals by Benameur and Oyono-Oyono in \cite{BenameurOyono07}; the full proof of the ``magnetic'' gap labelling conjecture was achieved by Benameur and Mathai in \cite{BenameurMathai15,BenameurMathai18}. In the case of periodic potentials, the proof of the $K$-theoretic reformulation of the magnetic gap labelling conjecture can be traced back to a theorem by Elliott \cite{Elliott84} on the $K$-theory of the rotation $C^*$-algebra, elaborating upon earlier results by Connes \cite{Connes80}, Pimsner and Voiculescu \cite{PimsnerVoiculescu80}, and Rieffel \cite{Rieffel81} for the two-dimensional case which is also of interest for the present paper.
Unfortunately, Bellissard used the term ``St{\v r}eda formula'' to denote the equality between the derivative of the integrated density of states and the Chern marker, although St{\v r}eda's contribution strictly consisted in relating the derivative with respect to the magnetic field of the integrated density of states with the Hall conductivity, within linear response. We note that the actual ``St{\v r}eda formula'', in the latter sense, was rigorously proved in the gapped continuum case by Cornean {\sl et al.}  \cite{CorneanNenciuPedersen06}.

Schulz-Baldes and Teufel  \cite{SchulzBaldesTeufel13} significantly improved the results of St{\v r}eda and also extended the results of \cite{Bellissard89} to the case when the Fermi energy is situated in a mobility gap (see also \cite{BellissardSchulz-Baldes94} for the proof of integrality of the Chern marker in the latter regime). Of a similar flavour are certain higher-dimensional generalizations of the ``St{\v r}eda formula'', presented in the monograph by Prodan and Schulz-Baldes \cite{ProdanSchulz-Baldes16}. We note though that their proofs are formulated for bounded Hamiltonians.

\subsection{Goals and structure}
In this paper, we first provide a proof of the Wannier diagrams for  unbounded Bloch-Landau Hamiltonians acting in $L^2(\R^2)$; we achieve this in Theorem \ref{thm:Main}\eqref{(i)} and \eqref{(ii)}. As a by-product, we show in Corollary \ref{crl:Main} that while the Fermi projection is everywhere continuous in the strong topology as a function of the magnetic flux, there are situations in which this map is nowhere continuous in the norm topology!

Our second novel result, Theorem \ref{thm:StredaLimSupInf}, extends the gap labelling (in a weaker sense) to more general perturbations; in particular, we are interested in perturbations given by slowly varying magnetic fields, which generically break covariance.  

The main tool we use is the so-called gauge covariant magnetic perturbation theory developed by Cornean and Nenciu \cite{CorneanNenciu98,Nenciu02,CorneanNenciu09,Cornean10}. We do not use noncommutative geometry, and clarify the physical meaning of the ``gap labelling'' through Wannier diagrams. 

The structure of the paper is as follows. In the rest of this Section we formulate our two main results, namely Theorems \ref{thm:Main} and \ref{thm:StredaLimSupInf}. In Section \ref{sect2} we prove Theorem \ref{thm:Main}, in Section \ref{sect3} we prove Theorem \ref{thm:StredaLimSupInf}, while in Appendix \ref{app:ChernAppendix} we review the magnetic Bloch-Floquet transform and the Chern number. Then, in Appendix \ref{app:Kernel} we review the gauge covariant magnetic perturbation theory and we end with Appendix \ref{app:KatoNagy}, in which we prove a localization estimate for the integral kernel of a Kato-Nagy unitary operator that will be used in the proof of Corollary \ref{crl:Main}.

\subsection{The covariant setting} 

We consider a Bloch-Landau Hamiltonian acting on $L^2(\R^2)$, defined in Hartree atomic units by
\[
H_b=\frac{1}{2}\left({\bf P} - b{\bf A}\right)^2 + V \, ,
\]
where ${\bf P}=-\iu \nabla$ is the usual momentum operator, $b \in \R$ is the magnetic field strength divided by the speed of light, and ${\bf A}=\frac{1}{2} (-x_2,x_1)$ is the magnetic potential in the symmetric gauge. $V$ is a $\Z^2$-periodic electric potential; although we could  handle  certain singularities, in order to streamline the proofs we choose to work with bounded $V$'s.   We note that we could work with any periodicity lattice $\Gamma \subset \R^2$ instead of $\Z^2$, but for simplicity  we only consider $\Z^2$-periodic potentials. 
Under these conditions, $H_b$ is essentially selfadjoint on $C^\infty_0(\R^2)$.

Suppose that the spectrum of $H_b$ has an {\em isolated} spectral island $\sigma_b$; by definition, isolated means that it is separated by the rest of the spectrum by one or two gaps. The edges of these gaps vary continuously with $b$ \cite{Nenciu86}, thus via the Riesz formula we can define the spectral projection
\[
\F_b=\frac{\iu}{2 \pi } \oint_{\mathcal{C}}  \left(H_b-z {\bf 1}\right)^{-1} dz\, 
\]
where $\mathcal{C}$ is a positively oriented simple contour encircling $\sigma_b$ and staying at a positive distance from the spectrum as long as the two gaps we started with remain open. We note that the nature and the structure of $\sigma_b$ can dramatically change when $b$ varies, i.e.\ internal mini-gaps can open or close, but as a set, $\sigma_b$ varies continuously with respect to the Hausdorff distance. It is possible to prove that $\F_b$ admits a jointly continuous integral kernel, see Appendix \ref{app:Kernel}. Using Combes-Thomas estimates \cite{CT73,CorneanNenciu09} (see also Appendix \ref{app:Kernel}) one can prove the existence of $\alpha,C>0$ such that  
\begin{equation}
\label{hc2}
|\F_b(\x;\x')|\leq C e^{-\alpha \Vert \x-\x'\Vert},\quad \forall \x,\x'\in \R^2.
\end{equation}
These two constants can be chosen uniformly in $b$ as long as $b$ varies in an interval such that {the distance between $\sigma_b$ and the rest of the spectrum is bounded from below by a positive constant. }

Let $\Lambda_L$ be the square of side-length $L>1$ centered at the origin and let $\chi_L$ be its characteristic function. Due to \eqref{hc2} we have that both $\chi_L \Pi_b e^{\alpha \|\cdot\|/2}$ and $e^{-\alpha \|\cdot\|/2}\Pi_b$ are Hilbert-Schmidt operators, thus
\begin{equation}
\label{hc3}
\chi_L \Pi_b=\left \{\chi_L \Pi_b e^{\alpha \|\cdot\|/2}\right\} \left\{ e^{-\alpha \|\cdot\|/2}\Pi_b\right\}
\end{equation}
is trace class with $0\leq   \Tr(\chi_L \Pi_b\chi_L)=\Tr(\chi_L \Pi_b)\leq C\, L^2$.

In the above considerations, the periodicity of $V$ played no role. When $V$ is $\Z^2$-periodic, we decompose every point $\x \in \R^2$ uniquely as
\[
\x=\gamma + \xx \, , \qquad \gamma \in \Z^2\, , \quad \xx \in \left(-1/2,1/2\right] \times \left(-1/2,1/2\right]   =: \Omega \, .
\]
Define the Peierls antisymmetric phase by
\begin{equation}
\label{Peierlsphase}
\phi(\x,\x'):=\frac{1}{2} ( x'_1 x_2-x_2'x_1 )\, , \quad  \forall \, \x,\x' \in \R^2\, .
\end{equation}
The phase satisfies the composition rule
\begin{equation}
\label{CompositionRule}
\phi(\x,\y) + \phi(\y,\x') = \phi(\x,\x') + \phi(\x-\y,\y-\x') \quad \forall \, \x, \y, \x' \in \R^2\,  .
\end{equation}
The Hamiltonian $H_b$ commutes with the magnetic translations $\tau_{b,\eta}$, defined for every $\eta \in \Z^2$ by
\[
(\tau_{b,\eta} \psi )(\x):= e^{\iu b \phi(\x,\eta)} \psi(\x-\eta) \, , \qquad \forall \psi \in L^2(\R^2)\,. 
\] 
Because $\Pi_b$ then also commutes with the magnetic translations, for every $\eta \in \Z^2$ we have
\begin{equation}
\label{IntertwMagTrs}
\F_b(\x,\x')=e^{\iu b \phi(\x,\eta)} \F_b(\x-\eta;\x'-\eta) e^{-\iu b \phi(\x',\eta)} \,.
\end{equation}

We can then define the integrated density of states for the projection $\F_b$ by the formula
\begin{equation}
\label{Ids}
\Ids(\F_b):=  \lim\limits_{L \to \infty} \frac{1}{|\Lambda_L|} \Tr(\chi_L\F_b)=\int_\Omega  \, \F_b(\xx;\xx) d\xx\, ,
\end{equation}
where the second equality is a consequence of  the $\Z^2$-periodicity of $\Pi_b(\x;\x)$, implied by \eqref{IntertwMagTrs}. 

Let $X_j$ be the multiplication operator by $x_j$. Due to \eqref{hc2} we have that the commutators $[X_j,\Pi_b]$ have continuous integral kernels given by $(x_j-x_j')\Pi_b(\x;\x')$ and which are exponentially localized near the diagonal.  Due to the \eqref{IntertwMagTrs} we see that these commutators also commute with the magnetic translations. Again by \eqref{IntertwMagTrs}, the integral kernel of $\iu \F_b \left[ \left[X_1,\F_b\right] , \left[X_2,\F_b\right] \right]$ is such that for all $\eta \in \Z^2$
\begin{equation*}
\left(\iu\F_b \left[ \left[X_1,\F_b\right] , \left[X_2,\F_b\right] \right]\right)(\x;\x')= e^{\iu b \phi(\eta,\x)} \left(\iu\F_b \left[ \left[X_1,\F_b\right] , \left[X_2,\F_b\right] \right]\right)(\x-\eta,\x'-\eta) e^{-\iu b \phi(\eta,\x')}.
\end{equation*} 
Inspired by the non-commutative Chern character used in \cite{Bellissard86,Bellissard89,BellissardSchulz-Baldes94}, by the local Chern marker defined in \cite{BiancoResta2011}, and in accordance with the Chern marker recently defined in a more general setting in \cite{MarcelliMonacoMoscolariPanati}, we define the Chern marker as
\begin{equation}
\label{GenChern}
\Chern(\F_b):=  2 \pi \int_\Omega  \, \left(\iu\F_b \left[ \left[X_1,\F_b\right] , \left[X_2,\F_b\right] \right]\right)(\xx;\xx) d\xx\, .
\end{equation}

We are now prepared to formulate our first main result, a gap labelling theorem for Bloch-Landau Hamiltonians.
\begin{theorem}
	\label{thm:Main}
	Assume that  $H_{b}$ has an isolated spectral island $\sigma_{b}$ which remains isolated and varies continuously in the Hausdorff distance as long as $b\in (b_1,b_2)$. Let $\F_b$ be its corresponding spectral projection. Then:
	\begin{enumerate}[leftmargin=*,label=(\roman*),ref=(\roman*)]
		\item \label{(i)} the map $(b_1,b_2)\ni b\mapsto  \Ids(\F_b)\in \R$ is continuously differentiable and
		\begin{equation}
		\label{StredaFormula}
		\frac{d\Ids(\F_b)}{db}= \frac{1}{2\pi}\Chern(\F_b)
		\end{equation}
		with $\Chern(\F_b)$ as in  \eqref{GenChern};
		
		\item \label{(ii)} the Chern marker is constant on $(b_1,b_2)$ and
		$$
		\Chern(\F_b)=c_1 \in \Z ;
		$$
		Moreover, there exists a rational number $c_0\in \mathbb{Q}$ such that 
		\begin{align}\label{hc1'}		
		\Ids(\F_b)=c_0+c_1\frac{b}{2\pi},\quad b\in (b_1,b_2).
		\end{align}
	\end{enumerate}
\end{theorem}

As a by-product of the above gap labelling theorem, we can obtain some additional information about how singular the magnetic perturbation is. To formulate the statement more precisely, we will need the notion of a \emph{localized Wannier-like basis}. These are particular orthonormal bases for the range of the Fermi projection: the definition we choose below is motivated by \cite{CorneanMonacoMoscolari}, where such localized Wannier-like bases are explicitly constructed for Fermi projections corresponding to magnetic Hamiltonians with rational flux, and certain irrational perturbations thereof. 
In the present context, we say that the Fermi projection $\Pi_b$ admits a localized Wannier-like basis if 
\begin{equation} \label{Wannier-like}
\Pi_b(\x;\x') = \sum_{j=1}^{M} \sum_{\gamma \in \Z^2} \psi_{j,\gamma}(\x) \, \overline{\psi_{j,\gamma}(\x')}, \quad \psi_{j,\gamma}(\x) := e^{\iu \theta(\x)} \, \tau_{b',\gamma} w_j(\x),  
\end{equation}
where $M \in \N$, $\theta \colon \R^2 \to \R$, $b' \in \R$, and the functions $w_j \in L^2(\R^2)$ are such that the vectors $\left\{\psi_{j,\gamma}\right\}_{1 \le j \le M, \gamma \in \Z^2}$ are orthonormal and
\begin{equation*}
\sup_{j} \int_{\x \in \R^2} |w_j(\x)|^2 e^{2 \alpha \|\x\| } \mathrm{d}\x < K,
\end{equation*}
for some $K,\alpha >0$. It is understood that the parameter $b'$ and the function $\theta$ are chosen so that the magnetic covariance \eqref{IntertwMagTrs} holds.

From the previous $L^2$ estimate we can extract an $L^\infty$ estimate on the Wannier-like functions. First consider that $\psi_{j,0}=\Pi_{b}\psi_{j,0}$; then by  using the Cauchy-Schwarz inequality together with \eqref{hc2}, we obtain that
\begin{equation} \label{WannierExpLoc}
|w_j(\xx + \gamma)| \le C e^{-\alpha \|\gamma\|}
\end{equation}
for some $C>0$ uniform in $j \in \set{1, \ldots, M}$ and $\xx \in \Omega$. Moreover, notice that \eqref{WannierExpLoc} is a pointwise estimate since the functions $\psi_{j,\gamma}$ admits a continuous representative. Indeed, the functions $\psi_{j,\gamma}$ belong to the domain of the Bloch-Landau Hamiltonian, which is included in the local Sobolev space $H_{\mathrm{loc}}^2(\R^2)$. Therefore, by a standard Sobolev embedding argument, one can show that all the functions in the domain of the Hamiltonian admit a continuous representative.

\begin{corollary} \label{crl:Main}
	Under the same assumptions as in Theorem \ref{thm:Main}, the map $(b_1,b_2)\ni b\mapsto \Pi_b$ is continuous in the strong topology.
	
	On the other hand, assume that for $b\in (b_1,b_2)$ either {\rm (i)} $c_1\ne 0$ in \eqref{hc1'}, or {\rm (ii)}  $c_1=0$ and $\Pi_b$ admits a localized Wannier-like basis. Then:
	\begin{align}\label{hc40}		
	\lim_{\epsilon\to 0}\| \F_{b+\epsilon}-\F_b\|=1.
	\end{align}
	In particular, \eqref{hc40} holds true if $c_1=0$ and $b/(2\pi)$ is rational.
\end{corollary}

\begin{remark}
	Corollary \ref{crl:Main} highlights the  singularity of the magnetic perturbation: the map $(b_1,b_2)\ni b\mapsto \Pi_b$ is continuous in the strong topology, but dramatically fails to be continuous in the norm topology. 
	
	While the case $c_1\neq 0$ is a rather straightforward consequence of Theorem \ref{thm:Main}, the case $c_1=0$ is more involved. A less general situation of this latter case was already treated by Nenciu  \cite[Lemma 5.8]{Nenciu91}. He only considered $b=0$ while $\sigma_0$ was a simple absolutely continuous band for which he assumed the existence of an orthonormal basis of exponentially localized Wannier functions. The strategy behind our proof is essentially the same as Nenciu's, but we generalize his argument in particular to any rational flux, by showing in Appendix \ref{app:ChernAppendix} that also in this case one can construct an orthonormal basis of localized Wannier functions for the Fermi projection.
	
	Motivated by magnetic perturbation theory \cite{CorneanMonacoMoscolari}, we conjecture that \eqref{hc40} and the existence of a localized Wannier-like basis as in \eqref{Wannier-like} should also hold for all irrational fluxes when $c_1=0$.  
\end{remark}

\subsection{Slowly varying magnetic perturbations}

Here we discuss the generalization of the above Diophantine formula \eqref{hc1'} to magnetic field perturbations that are slowly varying with respect to the lattice $\Z^2$, in the sense of space adiabatic perturbation theory.

Let $\ASl(\x)=(\ASl_1(x_1,x_2),\ASl_2(x_1,x_2))$ be a $C^2$ magnetic  potential and define $B:=\partial_2 \ASl_1-\partial_1 \ASl_2$. {Up to a simple gauge transformation, we may assume that $\ASl({\bf 0})={\bf 0}$.} Moreover, we assume that $B$ is at least $C^1$ with bounded derivatives in the following way:
\begin{equation}
\label{BoundB}
\sup_{\x \in \R^2} | \partial^\alpha B(\x) | \leq C_\alpha\, , \qquad \alpha\in \mathbb{N}^2,\quad |\alpha|\leq 1.
\end{equation}
{On top of that, we require that in the limit of large scales the magnetic field has a convergent flux per unit area, that is we assume the existence of the limit}
	\begin{equation}
	\label{PeridicLimitFlux}
	\langle B \rangle:=\lim\limits_{L \to \infty} \frac{1}{|\Lambda_L|} \int_{\Lambda_L} B(\x) d \x .
	\end{equation}
{Without loss of generality, we assume that $\langle B \rangle\geq 0$.}

Let $0< \lambda \ll 1$ denote the slow variation parameter. 
Let us introduce $\ASl_\lambda(\x):=\ASl(\lambda \x)$. Then $\ASl_\lambda$ produces a  slowly varying magnetic field $B_\lambda(\x):=\lambda B(\lambda \x)$.  
Let us consider the perturbed Hamiltonian of the form
\begin{equation}\label{hc20}
H_{b,\lambda}:=\frac{1}{2}\left({\bf P} - b{\bf A} + \ASl_\lambda \right)^2 + V \, ,
\end{equation}
with $b$, ${\bf A}$ and $V$ as before. Up to a gauge transformation, we may assume that $\ASl_\lambda$ is given in the transverse gauge: 
\begin{equation}\label{tr-g}
\ASl_\lambda(\x)= \left(\int_0^1 s B_\lambda (s\x)ds \right)\; (-x_2,x_1).
\end{equation}

$H_{b,\lambda}$ remains essentially selfadjoint on $C^\infty_{0}(\R^2)$. Like in the previous section, we assume that $H_{b,0}$ has an isolated spectral island $\sigma_{b,0}$. Since the perturbing magnetic field is of order $\lambda$, then for $\lambda$ small enough the perturbation given by $\ASl_\lambda$ does not close the gap between $\sigma_{b,0}$ and the rest of the spectrum \cite{Nenciu86,Nenciu02} (see also Appendix \ref{app:Kernel}). Thus $H_{b,\lambda}$ still has a spectral island $\sigma_{b,\lambda}$ ``close to'' $\sigma_{b,0}$. Via a Riesz integral we can define $\Pi_{b,\lambda}$ to be the spectral projection onto the spectral island $\sigma_{b,\lambda}$.

The operator $H_{b,\lambda}$ is not necessarily covariant anymore (i.e., it need not commute with some magnetic translations) and we can no longer be sure that $\Pi_{b,\lambda}$ admits an integrated density of states in the sense of \eqref{Ids}, namely the existence of the limit
\[ \lim_{L \to \infty} p^L_{b,\lambda}, \quad \text{where} \quad  p^L_{b,\lambda} := \frac{1}{|\Lambda_L|}  \Tr\left(\chi_L \Pi_{b,\lambda} \right) \]
is not always guaranteed. Nevertheless, the $\liminf$ and $\limsup$ of $p^L_{b,\lambda}$ always exist because the sequence is bounded in $L$ (see also \eqref{hc3}).

Now we are ready to state the second main result of our paper.

\begin{theorem}
	\label{thm:StredaLimSupInf}
	Let $\Pi_{b,\lambda}$ be the spectral projection defined above. Assume that the limit defining $\langle B \rangle$ in  \eqref{PeridicLimitFlux} exists. Denote by $I_\lambda$ either $\displaystyle\limsup_{L\to\infty} p^L_{b,\lambda}$ or $\displaystyle\liminf_{L\to\infty} p^L_{b,\lambda}$. Then
	\begin{align}
	\label{StredaLimSup}
	&I_\lambda = \Ids(\Pi_{b,0}) + \lambda \frac{\langle B \rangle}{2\pi} \Chern(\Pi_{b,0}) + \mathcal{O}(\lambda^2) \, .
	\end{align}
\end{theorem}

\begin{remark}
	Theorem \ref{thm:StredaLimSupInf} says that  even if the integrated density of states might not exist, the first order terms in $\lambda$ of $\limsup_L p^L_{b,\lambda}$ and $\liminf_L p^L_{b,\lambda}$ are equal and proportional to the Chern marker of the unperturbed projection, thus the possible failure in the existence of an integrated density of states is only quadratic in $\lambda$.  
\end{remark}

\section{Proof of Theorem \ref{thm:Main} and Corollary \ref{crl:Main}}\label{sect2}
\subsection{Proof of \eqref{(i)}} Let us fix some $b\in (b_1,b_2)$ and assume that $\epsilon\neq 0$ is such that $b+ \epsilon\in (b_1,b_2)$. Proving \eqref{StredaFormula} is equivalent to showing
\begin{equation} \label{expansion}
\Ids(\F_{b+\epsilon})= \Ids(\F_{b}) + \frac{\epsilon}{2\pi} \Chern(\F_b) +o(\epsilon)\, , \quad \epsilon \to 0 \, .
\end{equation}
It is well known in the literature \cite{CorneanNenciu98, Nenciu02} that the constant magnetic field induces a  singular perturbation. 
Fortunately, in order to compute $\Ids(\F_{b+\epsilon})$ we only need a good control on the diagonal value $\F_{b+\epsilon}(\x;\x)$ of the integral kernel. The gauge covariant magnetic perturbation theory provides us with a convergent expansion in $\epsilon$ of exactly such objects. 

First of all, we define the operator $\widetilde{\F}^{(\epsilon)}$ given by the following integral kernel:
\begin{equation}\label{hc4}
\widetilde{\F}^{(\epsilon)}(\x;\x') = e^{\iu \epsilon \phi(\x,\x')} \F_{b}(\x;\x') \, .
\end{equation}
Note that the operator $\widetilde{\F}^{(\epsilon)}$ is selfadjoint due to the antisymmetry of the Peierls phase defined in \eqref{Peierlsphase} and to the selfadjointness of $\Pi_b$.

Using the gauge covariant magnetic perturbation theory as in \cite{Nenciu02} (see also Appendix \ref{app:Kernel}) one can show that there exist two constants $\alpha,K>0$ such that 
\begin{align}\label{hc5}
\left|\F_{b+\epsilon}(\x;\x')-\widetilde{\F}^{(\epsilon)}(\x;\x') \right| \leq |\epsilon| K e^{-\alpha\|\x-\x'\|}\,.
\end{align}
In fact we could give an explicit formula for the difference in the left-hand side in all orders of $\epsilon$, but the expression is  complicated and contains contributions coming from all spectral subspaces of $H_{b}$, not just from the one corresponding to $\Pi_b$. Using such an exact formula in order to show that the first order contribution in $\epsilon$ to $\Ids(\Pi_{b+\epsilon})$  is proportional to $\Chern(\Pi_b)$ seems to be computationally involved and would surely demand the use of many rather obscure identities and sum rules. 

Instead, to prove \eqref{expansion} we will use a quite different strategy {\bf which only involves the integral kernel of $\Pi_b$, the knowledge that $\Pi_b$ is a projection,  and the a-priori zero-order estimate \eqref{hc5}}. This strategy consists of two steps:

\begin{enumerate}[leftmargin=*,label={{\it Step \arabic*.}}]
	\item {Using the fact that $\widetilde{\F}^{(\epsilon)}$ is an ``almost'' projection, we will explicitly construct an auxiliary ``true'' projection $\mathcal{P}^{(\epsilon)}$ which, for $|\epsilon|$ small enough, is unitarily equivalent to $\F_{b+\epsilon}$ through an unitary operator that satisfies the hypothesis of Lemma \ref{MainLemma} (see below). As a consequence, we will show that $\mathcal{P}^{(\epsilon)}$ has {\bf the same} integrated density of states as $\F_{b+\epsilon}$. }
	
	\item We will study the asymptotic behavior in $\epsilon$ of the integrated density of states of $\mathcal{P}^{(\epsilon)}$ and show that
	\begin{align}\label{hc6}
	\Ids(\mathcal{P}^{(\epsilon)})= \Ids(\mathcal{P}^{(0)}) + \frac{\epsilon}{2\pi} \Chern(\F_b) +o(\epsilon)\, , \quad \epsilon \to 0 \, .
	\end{align}	
\end{enumerate}

\subsubsection{Step 1} 
Define the operator  
\begin{equation*}
\Delta^{(\epsilon)}:=\big(\widetilde{\F}^{(\epsilon)}\big)^2-\widetilde{\F}^{(\epsilon)} \, .
\end{equation*}
The operator $\Delta^{(\epsilon)}$ measures how far $\widetilde{\F}^{(\epsilon)}$ is from being a projection.
Using \eqref{hc2} and \eqref{CompositionRule} one can prove (see \eqref{eqn:Delta} below and also \cite[Section 9.3]{CorneanMonacoMoscolari}) that if $\epsilon$ is small enough then
\begin{equation}
\label{DeltaKer}
|\Delta^{(\epsilon)}(\x;\x')| \leq  |\epsilon | K e^{-\alpha\|\x-\x'\|} \, ,
\end{equation}
where $K$ is a positive constant. Notice that in the following $K$ will denote a generic positive constant. 

Thus, for $\epsilon $ small enough, we can construct the following orthogonal projections (see also  \cite{Nenciu02} for more details):
\begin{align}\label{ProjEps1}
\mathcal{P}^{(\epsilon)}:=\widetilde{\F}^{(\epsilon)}+\big (\widetilde{\F}^{(\epsilon)}-\tfrac{1}{2} {\bf 1}\big )\big \{ ({\bf 1}+4\Delta^{(\epsilon)})^{-1/2}-{\bf 1}\big \} \, .
\end{align}
Since the integral kernel of $\Delta^{(\epsilon)}$ is exponentially localized and of order $\epsilon$, one can prove (see \cite[Lemma 8.5]{CorneanMonacoMoscolari}) that 
\begin{equation}
\label{Aux}
\left| \big \{ ({\bf 1}+4\Delta^{(\epsilon)})^{-1/2}-{\bf 1} +2 \Delta^{(\epsilon)}\big \}(\x;\x') \right| \leq \epsilon^2 K e^{-\alpha\|\x-\x'\|}
\end{equation}
This estimate, combined with definition \eqref{ProjEps1} and with \eqref{hc5}, yields the following pointwise estimate:
\begin{equation}
\label{GoodLoc}
\left|\left(\F_{b+\epsilon}-\mathcal{P}^{(\epsilon)}\right)(\x;\x') \right| \leq K |\epsilon|\;  e^{-\alpha\|\x-\x'\|}\, .
\end{equation}
Due to \eqref{GoodLoc} we have that $\|\F_{b+\epsilon}-\mathcal{P}^{(\epsilon)}\|\leq C |\epsilon|\leq 1/2$ when $|\epsilon|$ is sufficiently small, hence we can consider the Kato-Nagy unitary operator $\mathsf{U}_\epsilon$ \cite{Kato66} such that $\F_{b+\epsilon}=\mathsf{U}_\epsilon\mathcal{P}^{(\epsilon)}\mathsf{U}_\epsilon^*$. From its explicit expression one can obtain the following estimate (see \cite[Lemma 8.5]{CorneanMonacoMoscolari}, c.f. also Appendix \ref{app:KatoNagy}):
\begin{equation}
\label{RequirementForU}
\left|\left(\mathsf{U}_\epsilon-\Id\right)(\x;\x')\right| \leq  C e^{-\alpha\|\x-\x'\|} \, , 
\end{equation}
which holds	for some positive constants $C$ and $\alpha$, provided $|\epsilon|$ is small enough.

Now we prove that $\F_{b+\epsilon}$ and $\mathcal{P}^{(\epsilon)}$ have the same integrated density of states if $|\epsilon|$ is small enough. In order to do that, we use the following general lemma.

\begin{lemma}
	\label{MainLemma}
	Let $P_1$ and $P_2$ be two orthogonal projections such that their integral kernels satisfy \eqref{hc2}. Assume that there exists a unitary operator $U$ such that $P_1=U P_2 U^*$ and whose integral kernel satisfies \eqref{RequirementForU}.  
	
	Then we have
	\[
	\lim\limits_{L \to \infty} \frac{1}{|\Lambda_L|} \left \vert \Tr(\chi_L P_1)-\Tr(\chi_L P_2)\right \vert =0\, .
	\]
	{In particular, if one projection admits an integrated density of states as in \eqref{Ids}, then both of them do and $\Ids(P_1)=\Ids(P_2)$.}
\end{lemma}

\begin{proof} Reasoning as in \eqref{hc3} we observe that the operators $\chi_L P_1$, $\chi_L P_2$ and $\chi_L U P_2$ are trace class. {Exploiting the invariance of the trace under unitary conjugation,} we obtain the identity
	$$\Tr(\chi_L P_1)-\Tr(\chi_L P_2)=\Tr \left ( [\chi_L,U]P_2U^*\right ).$$
	Denoting by $W:=U-\Id$ we see from \eqref{RequirementForU} that $W(\x;\x')$ is exponentially localized near the diagonal and
	$$\Tr \left ( [\chi_L,U]P_2U^*\right )=\Tr \left ( [\chi_L,W]P_2\right )+\Tr \left ( [\chi_L,W]P_2 W^*\right ).$$
	Both traces can be bounded by a double  integral of the type 
	$$\int_{\x\in \R^2}\int_{{\x'}\in \R^2}e^{-\alpha \Vert \x-\x'\Vert}|\chi_L(\x)-\chi_L(\x')|d\x' d\x\; .$$
	In the above integral, the integrand is non-zero only if one variable belongs to $\Lambda_L$ and the other one lies outside $\Lambda_L$. Due to the symmetry, it is enough to estimate
	$$\int_{\x\in \Lambda_L}\int_{{\x'}\in \R^2\setminus \Lambda_L}e^{-\alpha \Vert \x-\x'\Vert}d\x' d\x\; .$$
	For a fixed $\x\in \Lambda_L$ we have the inequality
	$$e^{-\alpha \Vert \x-\x'\Vert}\leq e^{-\alpha\;  {\rm dist}(\x,\partial \Lambda_L)/2}e^{-\alpha \Vert \x-\x'\Vert /2},\quad \forall \x'\in \R^2\setminus \Lambda_L.$$
	By integrating with respect to $\x'$ at fixed $\x$ we can bound the above double integral by
	$$\int_{\x\in \Lambda_L}e^{-\alpha\;  {\rm dist}(\x,\partial \Lambda_L)/2}d\x\leq C \;L,$$
	hence when dividing by $L^2 = |\Lambda_L|$ we obtain the claimed convergence to zero. 
\end{proof}

Using \eqref{IntertwMagTrs} one can prove by direct computation that the operator $\widetilde{\Pi}^{(\epsilon)}$ commutes with the magnetic translations 
$\tau_{b+\epsilon,\eta}$. Since $\mathcal{P}^{(\epsilon)}$ is a function of $\widetilde{\Pi}^{(\epsilon)}$, it also commutes with the same magnetic translations, thus $\widetilde{\Pi}^{(\epsilon)}(\x;\x)$ and $\mathcal{P}^{(\epsilon)}(\x;\x)$ are periodic functions and the integrated densities of states $\Ids(\widetilde{\Pi}^{(\epsilon)})$, $\Ids(\mathcal{P}^{(\epsilon)})$ exist. Due to \eqref{GoodLoc} and \eqref{RequirementForU} we can apply Lemma~\ref{MainLemma} to $\mathcal{P}^{(\epsilon)}$ and $\Pi_{b+\epsilon} = \mathsf{U}_\epsilon \mathcal{P}^{(\epsilon)} \mathsf{U}_\epsilon^*$ and conclude that 
$$  
\Ids(\F_{b+\epsilon})=\Ids(\mathcal{P}^{(\epsilon)}).
$$
\subsubsection{Step 2} We now prove \eqref{hc6}. Let us begin by studying $\mathcal{P}^{(\epsilon)}$ in detail. Using the same method yielding \eqref{GoodLoc} but taking into account also the term of order $\epsilon$ we obtain the estimate 
$$
\left |\mathcal{P}^{(\epsilon)}(\x;\x')-\left \{\widetilde{\F}^{(\epsilon)}-2\widetilde{\F}^{(\epsilon)}\Delta^{(\epsilon)} + \Delta^{(\epsilon)}\right \}(\x;\x')\right |\leq C\epsilon^2 e^{-\alpha \Vert \x-\x'\Vert}\, .
$$
This leads to
\begin{equation}
\label{epsExpansion}
\left |\mathcal{P}^{(\epsilon)}(\x;\x)-\mathcal{P}^{(0)}(\x;\x)-\left \{-2\widetilde{\F}^{(\epsilon)}\Delta^{(\epsilon)} + \Delta^{(\epsilon)}\right \}(\x;\x)\right |\leq C\epsilon^2,
\end{equation}
where we used that $\widetilde{\F}^{(\epsilon)}(\x;\x)=\F_b(\x;\x)=\mathcal{P}^{(0)}(\x;\x)$, independent of $\epsilon$.

Exploiting the composition rule for the Peierls phase \eqref{CompositionRule}, the fact that $\Pi_b$ is a projection, and the exponential localization of the integral kernel of $\Pi_b$, we obtain
\begin{equation} \label{eqn:Delta}
\begin{aligned}
&\Delta^{(\epsilon)}(\x;\x')=\int_{\R^2} d \y \, \left( e^{\iu \epsilon \phi(\x,\y)} \F_{b}(\x;\y) e^{\iu \epsilon \phi(\y,\x')} \F_{b}(\y;\x') - e^{\iu \epsilon \phi(\x,\x')}  \F_{b}(\x;\y)  \F_{b}(\y;\x')  \right) \\
&=\frac{\iu}{2}  e^{\iu \epsilon \phi(\x,\x')} \epsilon\int_{\R^2} d \y \left[(\x-\y)_2 (\y-\x')_1 - (\x-\y)_1 (\y-\x')_2 \right]  \F_{b}(\x;\y) \F_{b}(\y;\x') \\
&\phantom{=}+ \mathcal{O}(\epsilon^2 e^{-\alpha \Vert \x-\x'\Vert}) \, .
\end{aligned}
\end{equation}
Noticing that $\Delta^{(\epsilon)}(\x;\x)=0 + \mathcal{O}(\epsilon^2)$ it follows that only 
$-2\widetilde{\F}^{(\epsilon)}\Delta^{(\epsilon)}(\x;\x)$ contributes to the first order expansion in $\epsilon$ of $\mathcal{P}^{(\epsilon)}(\x;\x)$ in  \eqref{epsExpansion}. More precisely
\begin{align*}
&-2\left(\widetilde{\F}^{(\epsilon)}\Delta^{(\epsilon)}\right)(\x;\x)= -2\int_{\R^2} d\tx\,  e^{\iu \epsilon \phi(\x,\tx)} \F_{b}(\x;\tx) \Delta^{(\epsilon)}(\tx;\x)  \\
&= \iu \epsilon \int_{\R^2} d\tx\, \F_{b}(\x;\tx)  \int_{\R^2} d \y \left[ (\tx-\y)_1 (\y-\x)_2-(\tx-\y)_2 (\y-\x)_1  \right]  \F_{b}(\tx;\y) \F_{b}(\y;\x) \\
&\phantom{=}+\mathcal{O}(\epsilon^2)\\
&=\epsilon\left(\iu\F_{b} \left[ \left[X_1,\F_{b}\right] , \left[X_2,\F_{b}\right] \right]\right)(\x;\x) + \mathcal{O}(\epsilon^2).
\end{align*}
This proves \eqref{hc6}, see \eqref{GenChern}. 

{The proof of the continuity of $\Chern(\F_b)$ as a function of $b\in (b_1,b_2)$ uses the same strategy and we only sketch it. First,  we replace $\F_{b+\epsilon}$ with $\widetilde{\Pi}^{(\epsilon)}$ in the expression of $\Chern(\F_{b+\epsilon})$ and using \eqref{hc5} we obtain
	$$
	\left|\Chern(\F_{b+\epsilon})-\Chern(\widetilde{\Pi}^{(\epsilon)})\right| \leq C |\epsilon| .
	$$
	Second, using the composition rule \eqref{CompositionRule} for the magnetic phases in the explicit expression of $\Chern(\widetilde{\Pi}^{(\epsilon)})$, we get
	$$
	\Chern(\widetilde{\Pi}^{(\epsilon)})=\Chern(\F_b) +\mathcal{O}(\epsilon)\, , \quad \epsilon \to 0 \, .
	$$
	Combining these two estimates, the continuity of $\Chern(\F_b)$ follows.}

\subsection{Proof of \eqref{(ii)}}

From Theorem \ref{thm:Main}\eqref{(i)} we know  that the derivative of the integrated density of states is a continuous function and is proportional to the Chern marker $\Chern(\F_{b})$ for every $b$ restricted to compact intervals in $(b_1,b_2)$ where the spectral island $\sigma_b$ remains isolated from the rest of the spectrum. 

The main observation, proved for the convenience of the reader in Appendix \ref{app:ChernAppendix}, is that the  map 
$$(b_1,b_2)\ni b\mapsto \Chern(\F_{b})\in \R$$
takes {\bf integer values} when $b/(2\pi)\in \mathbb{Q}$. Since the map is at the same time uniformly continuous on any compact interval included in $(b_1,b_2)$, a straightforward argument shows that it must be constant and thus everywhere equal to an integer $c_1\in \mathbb{Z}$. 

In order to prove \eqref{hc1'}, let us fix some $b_0\in (b_1,b_2)$ such that $b_0/(2\pi) =p/q\in \mathbb{Q}$. Then for every other $b$ in this interval we have by \eqref{StredaFormula}
$$\Ids(\Pi_b)=\Ids(\Pi_{b_0})+c_1 \frac{b}{2\pi} -\frac{c_1p}{q}.$$
In Appendix \ref{app:ChernAppendix} we will prove that $\Pi_{b_0}$ is a fibered operator. In the magnetic Bloch-Floquet representation, the fiber of $\Pi_{b_0}$ at a fixed quasimomentum ${\bf k}$ is a rank-$M$ orthogonal projection. Also, $\sigma_{b_0}$ is the union of $M$ mini-bands (which might overlap). When we compute $\Ids(\Pi_{b_0})$ with the help of \eqref{hc15}, the result is $M/q\in \mathbb{Q}$.  Thus setting $c_0:=(M-c_1p)/q$ concludes the proof.

\subsection{Proof of Corollary \ref{crl:Main}} 
The continuity of the function $b \mapsto \Pi_b$ with respect to the strong topology is known since at least Kato \cite{Kato66}, who used asymptotic perturbation theory. For the sake of the reader we present here a much shorter proof based on magnetic perturbation theory. By a standard density argument it is enough to show that 
$$\lim_{\epsilon\to 0}\|(\Pi_{b+\epsilon}-\Pi_b)\psi\|=0$$
for every $\psi$ with compact support. This limit follows from \eqref{hc2}, \eqref{hc4}, \eqref{hc5}, from the inequality
$$\left |e^{\iu \epsilon \phi(\x,\x')}-1\right |\leq |\epsilon|\; |\phi(\x,\x')|\leq  |\epsilon|\; \|\x-\x'\|\; \|\x'\| /2$$
and the fact that the map $\x' \mapsto \|\x'\|\; |\psi(\x')|$ belongs to $L^2(\R^2)$. 

Now let us continue with proving the discontinuity of the function $b \mapsto \Pi_b$ in the norm topology. We start with a general fact: if $P_1$ and $P_2$ are orthogonal projections, then $\|P_1-P_2\|\leq 1$ \cite[Chap.~I, Problem 6.33]{Kato66}. Hence in order to prove \eqref{hc40} it is enough to show that the $\liminf$ cannot be less than one.

Let $c_1\neq 0$. Assume that \eqref{hc40} is false. Then there would exist an $a \in [0,1)$ and a sequence $\epsilon_n\neq 0$, depending on $a$, such that $\epsilon_n\to 0$ and $\lim_{n\to \infty}\|\F_{b+\epsilon_n}-\F_b\|=a$. This implies the existence of some $n_0$ such that for every $n\geq n_0$ we have
$$\|\F_{b+\epsilon_n}-\F_b\|\leq \frac{(1+a)}{2}<1.$$
Then $\F_{b+\epsilon_n}$ and $\F_b$ would be intertwined by a Kato-Nagy unitary that satisfies the hypothesis of Lemma \ref{MainLemma}, see Lemma \ref{LemmaKNUnitary}. Thus $\mathcal{I}(\F_{b+\epsilon_n})=\mathcal{I}(\F_{b})$ if $n$ is large enough, which contradicts that $c_1\neq 0$.

Now let $c_1=0$, and assume \eqref{Wannier-like}. Let us define the unit vector 
$$\Psi_{\epsilon,\eta} (\x):= e^{\iu \epsilon \phi(\x,\eta)} \psi_{1,\eta}(\x),\quad \eta \in \Z^2.$$
Using \eqref{hc4}, \eqref{hc5}, and the exponential decay \eqref{WannierExpLoc} of $w_1$ we obtain the existence of $C>0$ such that for all $\eta$
\[
\langle \Psi_{\epsilon,\eta}, \Pi_{b+\epsilon}\Psi_{\epsilon,\eta}\rangle \geq 1-C\; |\epsilon|.
\]
Also
\[
\|\Pi_{b+\epsilon}-\Pi_b\|\geq \langle \Psi_{\epsilon,\eta},(\Pi_{b+\epsilon}-\Pi_b)\Psi_{\epsilon,\eta}\rangle \geq 1-C\; |\epsilon| -\sum_{j=1}^{M} \sum_{\gamma\in \Z^2} \left |\langle \Psi_{\epsilon,\eta},\psi_{j,\gamma}\rangle \right |^2.
\]
Since the left-hand side is independent of $\eta$ we have the inequality 
\begin{align}\label{hc44}
\|\Pi_{b+\epsilon}-\Pi_b\|\geq 1-C\; |\epsilon| -\inf_{\eta\in \Z^2}\sum_{j=1}^{M} \sum_{\gamma\in \Z^2} \left |\langle \Psi_{\epsilon,\eta},\psi_{j,\gamma}\rangle \right |^2.
\end{align}
We will now show that
\[
\lim_{|\eta|\to\infty }\sum_{j=1}^{M} \sum_{\gamma\in \Z^2} \left |\langle \Psi_{\epsilon,\eta},\psi_{j,\gamma}\rangle \right |^2=0,
\]
which inserted in \eqref{hc44} would finish the proof. 
By changing $\gamma$ into $\eta +\gamma$ we will investigate 
$$\sum_{\gamma\in \Z^2} \left |\langle \Psi_{\epsilon,\eta},\psi_{j,\gamma+\eta}\rangle \right |^2.$$
Due to the exponential localization of the $w_j$'s and using the triangle inequality one can prove the existence of two constants $\alpha,C>0$ such that 
$$\left |\langle \Psi_{\epsilon,\eta},\psi_{j,\gamma+\eta}\rangle \right |\leq Ce^{-\alpha \|\gamma\|},\quad \forall \eta \in \Z^2.$$
Thus the proof would be over if we can prove that for fixed $\gamma$ we have
\[
\lim_{|\eta|\to\infty } \langle \Psi_{\epsilon,\eta},\psi_{j,\gamma+\eta}\rangle =0.
\]

Let us compute
\[
\begin{aligned} \langle \Psi_{\epsilon,\eta},\psi_{j,\gamma+\eta}\rangle&= \langle e^{\iu \epsilon \phi(\cdot,\eta)} e^{\iu \theta(\cdot)} \, \tau_{b',\eta} w_1, e^{\iu \theta(\cdot)} \, \tau_{b',\gamma+\eta} w_j \rangle \\
& =  \langle e^{\iu \epsilon \phi(\cdot,\eta)} \, \tau_{b',\eta} w_1, e^{\iu b' \phi(\eta,\gamma)} \, \tau_{b',\eta} \tau_{b',\gamma} w_j \rangle,
\end{aligned} \]
where we used the fact that magnetic translations form a projective representation of $\Z^2$. An easy computation, exploiting $\phi(\eta,\eta)=0$, shows that multiplication by the phase factor $e^{\iu \epsilon \phi(\cdot,\eta)}$ commutes with the magnetic translation $\tau_{b',\eta}$. Up to a factor of modulus one, the above scalar product is then proportional to the integral
$$\int_{\R^2}e^{-\iu \epsilon\phi(\x,\eta) }\overline{w_1(\x)} e^{\iu b\phi([\x],\gamma) }w_j(\x-\gamma)d\x$$
where $[\x] \in \Z^2$ denotes the ``integer part'' in the decomposition $\x = \xx + [\x]$ with $\xx \in \Omega$. The above integral is proportional to the Fourier transform of the $L^1$ function 
$$\overline{w_1(\x)} e^{\iu b\phi([\x],\gamma) }w_j(\x-\gamma),$$ evaluated at the point $\xi=\frac{\epsilon}{2} (-\eta_2,\eta_1)$. Since $\gamma$ is fixed and $\epsilon\neq 0$, the Riemann-Lebesgue lemma implies that the integral goes to zero when $|\eta|\to\infty$. The proof is over. 

\section{Proof of Theorem \ref{thm:StredaLimSupInf}} \label{sect3} 

The strategy of the proof resembles that of Theorem~\ref{thm:Main}. In this section we denote by $\F_{\lambda} \equiv \Pi_{b,\lambda}$ the Fermi projection on the isolated spectral island $\sigma_{b,\lambda}$ of $H_{b,\lambda}$. We start by showing the existence of an auxiliary projection $\Fe^{(\lambda)}$, unitarily equivalent to $\F_{\lambda}$, which can be used to explicitly compute the first order expansion in $\lambda$ of $I_\lambda$  in Theorem \ref{thm:StredaLimSupInf}. 

Let us introduce the phase factor given by
\begin{align}\label{hc9}
\phi_\lambda(\x,\x'):=\int_{\x}^{\x'}\ASl_\lambda \equiv \int_{0}^{1}  \, \ASl_\lambda(\x'+s(\x-\x')) \cdot (\x-\x') \,ds .
\end{align}
Note that when $\ASl_\lambda$ (see \eqref{hc20}) comes from a constant magnetic field, we obtain the usual Peierls phase \eqref{Peierlsphase}.

Using results from magnetic perturbation theory \cite{Nenciu02} (see also Appendix \ref{app:Kernel}) we have 
\begin{equation}\label{hc30}
|\Pi_{\lambda} (\x;\x') - e^{\iu \phi_\lambda (\x,\x')} \Pi_0(\x;\x')  | \leq C \lambda e^{-\beta \|\x-\x'\|} \, .
\end{equation}

As before, we define the operator $\widetilde{\F}_{\lambda}$ through its integral kernel:
$$
\widetilde{\F}_{\lambda}(\x;\x'):=e^{\iu \phi_\lambda (\x,\x')} \Pi_0(\x;\x'),
$$
where $\widetilde{\F}_{\lambda}$ is selfadjoint due to the antisymmetry of the phase factor defined in \eqref{hc9}. 
We also define the auxiliary projection $\Fe^{(\lambda)}$ (the analogue of $\mathcal{P}^{(\epsilon)}$ from the previous section) as
\begin{align}\label{ProjEps}
\Fe^{(\lambda)}:=\widetilde{\F}_{\lambda}+\big (\widetilde{\F}_{\lambda}-\tfrac{1}{2} {\bf 1}\big )\big \{ ({\bf 1}+4 \Delta^{(\lambda)})^{-1/2}-{\bf 1}\big \} \, , \quad \Delta^{(\lambda)} := \widetilde{\F}_{\lambda}^2-\widetilde{\F}_{\lambda} \, ,
\end{align}
such that
\begin{equation}
\label{GoodLoc2}
\left|\left(\F_{\lambda}-\Fe^{(\lambda)}\right)(\x;\x')\right| \leq C \lambda e^{-\alpha\|\x-\x'\|}\, .
\end{equation}
From this, one can prove  \cite{CorneanMonacoMoscolari} that if $\lambda$ is small enough, then one can construct the Kato-Nagy unitary $\mathsf{U}^{(\lambda)}$ such that $\F_{\lambda}=\mathsf{U}^{(\lambda)} \Fe^{(\lambda)} \mathsf{U}^{(\lambda)*}$ and moreover (see \cite[Lemma 8.5]{CorneanMonacoMoscolari})
\begin{equation}
\label{ULoc}
\left|(\mathsf{U}^{(\lambda)}- \Id) (\x;\x') \right| \leq C e^{-\alpha\|\x-\x'\|} \, . 
\end{equation}

Now we are ready to prove equation \eqref{StredaLimSup}. We only show the proof for the $\limsup$ case since the $\liminf$ case is completely analogous.

The operator $\chi_L\Fe^{(\lambda)}$ is trace class (cf.\ \eqref{hc3}) and we have the trivial identity
$$\frac{1}{|\Lambda_L|}  \Tr\left(\chi_L \F_{\lambda} \right)= \frac{1}{|\Lambda_L|}  \left(\Tr\left(\chi_L \F_{\lambda}\right) -  \Tr\left(\chi_L\Fe^{(\lambda)}\right)  \right) +  \frac{1}{|\Lambda_L|} \Tr\left(\chi_L\Fe^{(\lambda)}\right).$$

Thanks to Lemma \ref{MainLemma}, the first term on the right-hand side of the above identity converges to zero as $L \to \infty$, hence taking the $\limsup$ of both sides yields
\begin{equation}
\label{limsup}
I_\lambda=\limsup_{L \to \infty} \frac{1}{|\Lambda_L|} \Tr\left(\chi_L\Fe^{(\lambda)}\right)   \, .
\end{equation}
What we have to prove now is that
\begin{equation}
\label{EquivalentStredaLS}
\limsup_{L \to \infty} \frac{1}{|\Lambda_L|} \Tr\left(\chi_L\Fe^{(\lambda)}\right) = \Ids(\F_0) + \lambda \frac{\langle B \rangle}{2\pi} \Chern(\F_0) + \mathcal{O}(\lambda^2) \, .
\end{equation}

As we have done in the case of a constant magnetic field, we need to study the expansion in $\lambda$ of the trace on the left-hand side of \eqref{EquivalentStredaLS} using \eqref{ProjEps},  and control the behaviour at large $L$. We separately analyse each term of \eqref{ProjEps}.

Denote by $\fl(\x,\y,\x')$ the magnetic flux generated by the slowly varying magnetic perturbation through the triangle $\langle \x,\y,\,\x' \rangle$ with corners situated at $\x$, $\y$ and $\x'$:
\begin{equation}
\label{Magneticflux}
\fl(\x,\y,\x')=\phi_\lambda(\x,\y)+\phi_\lambda(\x,\y)-\phi_\lambda(\x,\x') = \int_{\langle \x,\y,\,\x' \rangle }  \; \lambda  B(\lambda \tx) \, d \tx.
\end{equation}
Since $B$ has uniformly bounded derivatives (see \eqref{BoundB}), we obtain
\begin{equation}
\label{BoundMagneticFlux}
|\fl(\x,\y,\x')| \leq \lambda  C_{B}  \|\x-\y\| \|\y-\x'\|
\end{equation}
with $ C_{B}$ a positive constant that only depends on the magnetic field $B$. Using equations \eqref{GoodLoc2}, \eqref{Magneticflux}, and the fact that $\Fe_\lambda$ is a projection we obtain
\begin{align*}
\Delta^{(\lambda)}(\x;\x')&=\int_{\R^2}  e^{\iu \phi_\lambda(\x,\y)} \F_{\lambda}(\x;\y) e^{\iu  \phi_\lambda(\y,\x')} \F_{\lambda}(\y;\x')d \y - e^{\iu  \phi_\lambda(\x,\x')}  \F_{\lambda}(\x;\x') \\
&= e^{\iu  \phi_\lambda(\x,\x')} \int_{\R^2}  e^{\iu \fl(\x,\y,\x') } \F_{\lambda}(\x;\y) \F_{\lambda}(\y;\x')d \y - e^{\iu  \phi_\lambda(\x,\x')}  \F_{\lambda}(\x;\x')  \\
&=  e^{\iu  \phi_\lambda(\x,\x')} \iu \int_{\R^2}  \fl(\x,\y,\x')  \F_{\lambda}(\x;\y) \F_{\lambda}(\y;\x')d \y + \mathcal{O}(\lambda^2 e^{-\alpha \|\x-\x'\|}) \, .
\end{align*}
Given two vectors $\x$ and $\y$ we denote by $\{\x\wedge \y\}:=x_1y_2-x_2y_1$. From \eqref{Magneticflux} we have 
\begin{equation}
\label{AuxFlux}
\fl(\x,\y,\x')   
=  \frac{\lambda}{2} B(\lambda \x') \{(\x-\y)\wedge(\y-\x')\} + \int_{\langle \x,\y,\,\x' \rangle} \,  \, \lambda \, \big( B(\lambda \tx) - B(\lambda \x')\big ) d \tx   \, . 
\end{equation}
From \eqref{BoundB} we deduce that $B$ is a Lipschitz function:
\begin{equation}
\label{LipB}
|B(\x)-B(\x')| \leq K_{B} \|\x-\x'\|,\quad  \forall \x, \x' \in \R^2.
\end{equation}
Using the above estimate, the fact that the diameter of a triangle is less than the sum of the lengths of any two of its sides, and knowing that the area of the triangle is less than the product of the same two side-lengths, we get
\begin{equation}
\label{AuxFlux2}
\left|\int_{\langle \x,\y,\,\x' \rangle}   \, \lambda \, \left( B(\lambda \tx) - B(\lambda \x')   \right) \, d \tx\right| \leq \lambda^2 K_{B} \left(\|\x-\y\|^2\|\y-\x'\|+\|\x-\y\|\|\y-\x'\|^2\right)  \,.
\end{equation}
Therefore, exploiting \eqref{AuxFlux}, \eqref{AuxFlux2} and the exponential localization of the integral kernel of $\F_{\lambda}$, we obtain
\begin{equation*}
\label{eqAux}
\Delta^{(\lambda)}(\x;\x')= \iu \lambda  \int_{\R^2} d \y \,  B (\lambda \x') \frac{1}{2}\{(\x-\y)\wedge(\y-\x') \} \F_{\lambda}(\x;\y) \F_{\lambda}(\y;\x') + \mathcal{O}(\lambda^2e^{-\alpha \|\x-\x'\|})\, .
\end{equation*}
Putting $\x=\x'$ in the above equation we see that $\Delta^{(\lambda)}(\x;\x)=\mathcal{O}(\lambda^2)$ thus $\Delta^{(\lambda)}$ gives no contributions of order zero or $\lambda$ to  $|\Lambda_L|^{-1}\Tr\left(\chi_L\Fe^{(\lambda)}\right)$, uniformly in $L\geq 1$ (cf.\ the argument below \eqref{epsExpansion}).

For the next term in the expansion \eqref{ProjEps} we have
\begin{align*}
&-2\left(\Fe^{(\lambda)}\Delta^{(\lambda)}\right)(\x;\x) \\
&= \iu \lambda B(\lambda \x )\int_{\R^2} d\tx\, \F_{0}(\x;\tx) \cdot \\
&\phantom{= \iu \lambda} \cdot \int_{\R^2} d \y \left[ (\tx-\y)_1 (\y-\x)_2-(\tx-\y)_2 (\y-\x)_1  \right]  \F_{0}(\tx;\y) \F_{0}(\y;\x)+\mathcal{O}(\lambda^2)\\
&=\lambda B(\lambda\x) \left(\iu\F_{0} \left[ \left[X_1,\F_{0}\right] , \left[X_2,\F_{0}\right] \right]\right)(\x;\x) + \mathcal{O}(\lambda^2) \, .
\end{align*}
Thus we need to understand the behaviour of 
$$
\limsup_{L \to \infty} \frac{1}{|\Lambda_L|} \int_{ \Lambda_L}  B(\lambda\x) \left(\iu\F_{0} \left[ \left[X_1,\F_{0}\right] , \left[X_2,\F_{0}\right] \right]\right)(\x;\x) d\x\, .
$$ 
Because the integrand is uniformly bounded, it is enough to consider integer values for $L$, and in order to simplify the notation we assume that $L=2L'+1$ with $L'\in \mathbb{N}$. In this case we have 
$$\Lambda_L=\{\x=\underline{x}+\gamma,\quad \underline{x}\in \Omega,\; |\gamma_j|\leq L'\}.$$
Let us denote by $\mathfrak{C}(\x):=\left(\iu\F_{0} \left[ \left[X_1,\F_{0}\right] , \left[X_2,\F_{0}\right] \right]\right)(\x;\x)$. We have that $\mathfrak{C}(\xx+\gamma)=\mathfrak{C}(\xx)$ for every $\xx \in \Omega$ and $\gamma \in \Z^2$. Moreover
\begin{align*}
\int_{ \Lambda_L}  B(\lambda\x) & \mathfrak{C}(\x)\ d\x=
\sum_{|\gamma_j|\leq  L'} \int_{\Omega}  \, B(\lambda\xx+\lambda\gamma) \mathfrak{C}(\xx)d\xx\\
&=  \sum_{|\gamma_j|\leq L'} \int_{\Omega}  \, B(\lambda\gamma) \mathfrak{C}(\xx)d\xx+ \sum_{|\gamma_j|\leq L'} \int_{\Omega}  \, \left(B(\lambda\xx+\lambda\gamma) - B(\lambda\gamma)\right) \mathfrak{C}(\xx) d\xx\\
&= \frac{\Chern(\F_{0}) }{2\pi} \sum_{|\gamma_j|\leq L'} B(\lambda\gamma) + \lambda \, \mathcal{O}(L^2) 
\end{align*}
by \eqref{LipB}.  Therefore, in view of \eqref{PeridicLimitFlux}, in order to complete the proof it suffices to show that 
\begin{equation} \label{dm1}
\limsup_{L \to \infty} \left  |\frac{1}{|\Lambda_{L}|} \sum_{|\gamma_j|\leq L'} B(\lambda\gamma) -  \frac{1}{|\Lambda_{\lambda L}|} \int_{\Lambda_{\lambda L}}  \, B(\x) d \x \right |= \mathcal{O}(\lambda)\, .
\end{equation}
This is a consequence of the formula $|\Lambda_L|=\lambda ^{-2} |\Lambda_{\lambda L}|$ and of \eqref{LipB}. Indeed a computation similar to the above yields
$$ \lambda^2 B(\lambda \gamma)=\int_{|x_j-\lambda \gamma_j|\leq \lambda/2}B(\x)d\x +\mathcal{O}(\lambda^3),$$
which gives \eqref{dm1} upon summing over $\gamma \in \Lambda_L \cap \Z^2$. In turn, using \eqref{PeridicLimitFlux}, the estimate \eqref{dm1} can be rewritten as 
$$
\limsup_{L \to \infty} \frac{1}{|\Lambda_{L}|} \sum_{|\gamma_j|\leq L'} B(\lambda\gamma) =  \langle B \rangle+ \mathcal{O}(\lambda)\, .
$$
This ends the proof of Theorem \ref{thm:StredaLimSupInf}.


\appendix\

\section{Bloch-Floquet(-Zak) transform and the Chern number}
\label{app:ChernAppendix}
In this appendix we discuss the magnetic Bloch-Floquet transform and the Chern marker. The discussion is adapted to the special class of integral operators we work with. 

Let  $b_0=2\pi p/q$ for some $p,q$ co-prime integer numbers and define the (modified) magnetic translation of vector $\eta \in \Z^2$, $\widehat\tau_{b_0,\eta}$, to be the following unitary operator:
\begin{equation}
\label{ModifiedMagneticTranlsation}
(\widehat\tau_{b_0,\eta}f)(\x):= e^{\iu b_0 \eta_1\eta_2/2} e^{\iu b_0\phi(\x,\eta)} f(\x-\eta) \, ,  \qquad \forall f \in  L^2(\R^2)\, .
\end{equation} 
By direct computation one can prove that the set $\{\widehat\tau_{b_0,\gamma}\}_{\gamma \in \Z^2}$ forms a unitary projective representation of the group $\Z^2$, that is 
$$
\widehat\tau_{b_0,\gamma} \widehat\tau_{b_0,\xi} = e^{-\iu b_0 \gamma_2 \xi_1} \widehat\tau_{b_0,\gamma+\xi}\, , \qquad \forall\, \gamma,\xi \in \Z^2 \,.
$$ 
Considering the enlarged lattice 
$$
\Z^2_{(q)}:= \left\{ \eta \in \Z^2 \; | \;\eta=(\gamma_1,q\gamma_2)\; ,\gamma \in \Z^2 \right\} \, , 
$$
we have that $\{\widehat\tau_{b_0,\eta}\}_{\eta \in \Z^2_{(q)}}$ is a true unitary representation of $\Z^2_{(q)}$, that is
$$
\widehat\tau_{b_0,\eta} \widehat\tau_{b_0,\rho} =  \widehat\tau_{b_0,\eta+\rho} \, , \qquad \forall \, \eta,\rho \in \Z^2_{(q)} \, .
$$
Let us denote by $\Z^{2*}_{(q)}$ the dual lattice of $\Z^2_{(q)}$ and by $\mathbb{B}_{(q)}$ and $\Omega_{(q)}$ the unit cells of  $\Z^{2*}_{(q)}$ and $\Z^2_{(q)}$ respectively, i.e.
$$
\mathbb{B}_{(q)}:=\left(-\pi,\pi\right] \times \left(-\pi/q ,\pi/q \right], \quad \Omega_{(q)}:=(-1/2,1/2]\times(-q/2, q/2] 
.$$ 
$\mathbb{B}_{(q)}$ is usually called the (magnetic) Brillouin zone. 
We introduce the Bloch-Floquet unitary (denoted by $\BF$) as the operator which maps $L^2(\R^2)$ onto $ \int_{\mathbb{B}_{(q)}}^\oplus L^2(\Omega_{(q)})d\kk$ and acts on $f \in C^\infty_{0}(\R^2)$ as
\[
(\BF f)(\kk,\yy):= \frac{1}{|\mathbb{B}_{(q)}|^{1/2}} \sum_{\gamma \in \Z^2_{(q)}} e^{-\iu \kk \cdot \gamma} \, (\widehat\tau_{b_0,-\gamma}f)(\yy) \, ,\quad \kk \in \mathbb{B}_{(q)}\, ,\quad \yy \in \Omega_{(q)} \, .
\]
Its adjoint acts in the following way: 
\begin{align}\label{hc41}
(\BF^*\psi)(\yy+\eta)=\frac{1}{|\mathbb{B}_{(q)}|^{1/2}} \int_{\mathbb{B}_{(q)}}e^{\iu \kk \cdot \eta} e^{-\iu b_0 \eta_1\eta_2/2} e^{\iu b_0\phi(\yy,\eta)}\psi(\kk,\yy)d\kk.
\end{align}
Assume that $T$ is a bounded operator on $L^2(\R^2)$ with a jointly continuous integral kernel $T(\x;\x')$ for which there exists $\alpha,C>0$ such that 
$$|T(\x;\x')|\leq C\; e^{-\alpha\| \x-\x'\|},\quad \forall \x,\x'\in \R^2.$$
We also assume that $T$ commutes with 
the magnetic translations \eqref{ModifiedMagneticTranlsation} which leads to (see also \eqref{IntertwMagTrs})
\begin{equation*}
T(\x;\x')=e^{\iu b_0 \phi(\x,\eta)} T(\x-\eta;\x'-\eta) e^{-\iu b_0 \phi(\x',\eta)} \,, \quad \forall \, \eta \in \Z^2_{(q)} \, ,
\end{equation*}
or, by replacing $\x'$ with $\yy+\eta$ and $\x$ by $\xx+\gamma$,
\[
T(\xx+\gamma;\yy+\eta)=e^{\iu b_0 \phi(\gamma,\eta)}e^{\iu b_0 \phi(\xx,\eta)} T(\xx+\gamma-\eta;\yy) e^{-\iu b_0 \phi(\yy,\eta)} \,,\quad \forall \xx,\yy\in \Omega_{(q)}.
\]
Then a straightforward computation shows that $\BF T\BF^*$ is a fibered operator $\int_{\mathbb{B}_{(q)}}^\oplus t_\kk d\kk$ where $t_\kk$ is bounded on $L^2(\Omega_{(q)})$ and has the jointly continuous integral kernel
\begin{equation}\label{hc12}
t_\kk(\xx;\yy):= \sum_{\eta \in \Z^2_{(q)}}e^{-\iu \kk\cdot \eta}e^{-\iu b_0 \eta_1\eta_2/2} e^{-\iu b_0\phi(\xx,\eta)} T(\xx+\eta; \yy),\quad \forall x,y\in \Omega_{(q)}.
\end{equation}
We observe that the above kernel is $\Z^{2*}_{(q)}$-periodic in $\kk$ and its Fourier coefficients give us back the original kernel:
\[
T(\xx+\eta; \yy)=\frac{1}{|\mathbb{B}_{(q)}|} \int_{\mathbb{B}_{(q)}} t_\kk(\xx;\yy) e^{\iu \kk\cdot \eta}e^{\iu b_0 \eta_1\eta_2/2} e^{\iu b_0\phi(\xx,\eta)}d\kk,\quad \forall \xx,\yy\in \Omega_{(q)}.
\]
In particular: 
\begin{equation}\label{hc14}
\frac{1}{|\Omega_{(q)}|}\int_{\Omega_{(q)}} T(\xx;\xx)d\xx=\frac{1}{4\pi^2} \int_{\mathbb{B}_{(q)}} \left (\int_{\Omega_{(q)}} t_\kk(\xx;\xx)d\xx\right ) d\kk.
\end{equation}
Most importantly, if $T$ is an orthogonal projection like $\Pi_b$ in \eqref{IntertwMagTrs} with $b=b_0$, then its corresponding fiber denoted by $p_\kk$ is also an orthogonal projection, real analytic and periodic in $\kk$, with finite (and constant) rank, and \eqref{hc14} can be restated as:
\begin{equation}\label{hc15}
\lim_{L\to\infty}\frac{1}{|\Lambda_L|}{\rm Tr}_{L^2(\R^2)}(\chi_L \Pi_{b_0})=\frac{1}{4\pi^2} \int_{\mathbb{B}_{(q)}} {\rm Tr}_{L^2(\Omega_{(q)})}(p_\kk) d\kk =\frac{{\rm rank}(p)}{q}\in \mathbb{Q}.
\end{equation}

Next we study the operator 
$$T=\iu\; \Pi_{b_0} [[X_1,\Pi_{b_0}],[X_2,\Pi_{b_0}]]$$
which appears in \eqref{GenChern}. 
The commutators $[X_j,\Pi_{b_0}]$ have kernels  given by $(x_j-x_j')\Pi_{b_0}(\x;\x')$ and thus they are exponentially localized around the diagonal and commute with the magnetic translations. Let us find the fiber of $[X_j,\Pi_{b_0}]$.

Denote by $U^Z$ the fibered unitary operator acting on $\int_{\mathbb{B}_{(q)}}^\oplus L^2(\Omega_{(q)})d\kk$ given by the fiber 
$$(u_\kk^Z \psi)(\xx):= e^{-\iu \kk \cdot \xx} \psi(\xx),\quad \forall \psi\in L^2(\Omega_{(q)}).$$
The Zak modification of the Bloch-Floquet unitary is $\BFZ:= U^Z\BF$, and it will be called the Bloch-Floquet-Zak (BFZ) transform.
The integral kernel of the BFZ transform applied to $\Pi_{b_0}$ can be read off from \eqref{hc12}:
\[
p_\kk^Z(\xx;\yy):= \sum_{\eta \in \Z^2_{(q)}}e^{-\iu \kk\cdot (\xx+\eta-\yy)}e^{-\iu b_0 \eta_1\eta_2/2} e^{-\iu b_0\phi(\xx,\eta)} \Pi_{b_0}(\xx+\eta; \yy),\quad \forall \xx,\yy\in \Omega_{(q)}.
\]
Differentiating with respect to $k_j$ and conjugating back with $(u_\kk^Z)^*$ we obtain that the fiber of $[X_j,\Pi_{b_0}]$ in the Bloch-Floquet representation is
$$(u_\kk^Z)^*\; \iu \big (\partial_{k_j}p_\kk^Z \big ) \; u_\kk^Z.$$
Thus the Bloch-Floquet fiber of $T$ becomes
$$t_\kk =-\iu\,  p_\kk (u_\kk^Z)^*[\partial_{k_1}p_\kk^Z,\partial_{k_2}p_\kk^Z]u_\kk^Z.$$
Introducing this into \eqref{hc14} and using trace cyclicity we obtain 

\[
\int_{\Omega} T(\x;\x)d\x=\frac{1}{|\Omega_{(q)}|}\int_{\Omega_{(q)}} T(\xx;\xx)d\xx=-\frac{\iu }{4\pi^2} \int_{\mathbb{B}_{(q)}} {\rm Tr}_{L^2(\Omega_{(q)})}\big ( p_\kk^Z \;[\partial_{k_1}p_\kk^Z,\partial_{k_2}p_\kk^Z]\big )d\kk.
\]
Thus 
$$ 2\pi \lim_{L\to\infty}\frac{1}{|\Lambda_L|}{\rm Tr}_{L^2(\R^2)}(\chi_L T)=\frac{1}{2\pi \iu} \int_{\mathbb{B}_{(q)}} {\rm Tr}_{L^2(\Omega_{(q)})}\big ( p_\kk^Z \;[\partial_{k_1}p_\kk^Z,\partial_{k_2}p_\kk^Z]\big )d\kk=:\Chern(p^Z).$$
After an elementary but long computation one may show that 
$$
\Tr_{L^2(\Omega_{(q)})}\big ( p_\kk^Z\;d p_\kk^Z \wedge d p_\kk^Z\big ) - \Tr_{L^2(\Omega_{(q)})}\big ( p_\kk\; d p_\kk \wedge d p_\kk\big ) = d\left \{ \Tr_{L^2(\Omega_{(q)})}(p_\kk\;  (u_\kk^Z)^* \wedge d u_\kk^Z)\right \} \,.
$$ 
The right-hand side is periodic in $\kk$, therefore after integration on the Brillouin zone and an application of  Stokes' Theorem we obtain that $\Chern(p^Z)=\Chern(p)$. The latter is well-known to be an integer from the theory of vector bundles: for a direct proof (showing that it equals the winding number of the determinant of a certain smooth and $2\pi$-periodic unitary matrix) see \cite[Proposition 5.3]{CorneanMonacoMoscolari}.  More about the number $\Chern(p^Z)$ can be found e.g.\ in  \cite{Panati07}.

In particular, when $\Chern(p)=0$ we may find \cite{CorneanMonacoMoscolari} an orthonormal basis $\{\xi_{j}(\kk,y)\}_{j=1}^{{\rm rank}(p)}$ in the range of $p_\kk$ which consists of real analytic vectors in $\kk$ and which are also periodic. Applying the inverse Bloch-Floquet transform as in \eqref{hc41} we obtain exponentially localized Wannier vectors
$$w_j(\yy+\eta):=\frac{1}{|\mathbb{B}_{(q)}|^{1/2}} \int_{\mathbb{B}_{(q)}}e^{\iu \kk \cdot \eta} e^{-\iu b_0 \eta_1\eta_2/2} e^{\iu b_0\phi(y,\eta)}\xi_j(\kk,\yy)d\kk,\quad  \yy\in \Omega_{(q)},\quad \eta\in \Z^2_{(q)},$$
such that
$$\Pi_{b_0}(\x;\x')=\sum_{j=1}^{{\rm rank}(p)} \sum_{\gamma\in \Z^2_{(q)}} (\widehat{\tau}_{b_0,\gamma} w_j)(\x) \overline{(\widehat{\tau}_{b_0,\gamma} w_j)(\x')} =\sum_{j=1}^{{\rm rank}(p)} \sum_{\gamma\in \Z^2_{(q)}} (\tau_{b_0,\gamma} w_j)(\x) \overline{(\tau_{b_0,\gamma} w_j)(\x')} .$$
Notice that the above is exactly in the form \eqref{Wannier-like} in the statement of Corollary \ref{crl:Main}.

\section{Kernel regularity, exponential localization and gauge covariant magnetic perturbation theory}
\label{app:Kernel}
In this appendix we sketch the main ideas behind the estimates \eqref{hc5} and  \eqref{hc30} and collect all the regularity results on integral kernels that we have used in the proofs, directly or indirectly. We only focus on \eqref{hc30} because \eqref{hc5} is nothing but  \eqref{hc30} when the magnetic field perturbation vanishes.

Assume that the total magnetic field is given by $b+B_\lambda(\x)$ where 
\[
\sup_{\x\in\R^2}|\partial^\alpha B_\lambda(\x)|\leq \lambda^{|\alpha|+1} C_\alpha,\quad \alpha\in \mathbb{N}^2,\quad |\alpha|\leq 1. 
\]
Define the family of vector potentials depending on the parameter $\y\in\R^2$:
$$\mathcal{A}_\lambda (\x,\y):= \left( \int_0^1 s B_\lambda(\y+s(\x-\y))ds \right) \; (-x_2+y_2,x_1-y_1).$$
We have the estimates 
\begin{align}\label{hc21'}
|\partial_\x^\alpha \mathcal{A}_\lambda (\x,\y)|\leq \lambda^{|\alpha|+1} C_\alpha \;\|\x-\y\|,\quad \alpha\in \mathbb{N}^2,\quad |\alpha|\leq 1. 
\end{align}
It turns out that they all generate the same magnetic field $B_\lambda(\x)$. Denote by $\mathcal{A}_\lambda (\x):=\mathcal{A}_\lambda (\x,{\bf 0})$, as in \eqref{tr-g}. Then we must have that $\mathcal{A}_\lambda (\x)$ and $\mathcal{A}_\lambda (\x,\y)$ differ by a gradient, and one can show that 
$$\mathcal{A}_\lambda (\x)-\mathcal{A}_\lambda (\x,\y)=\nabla_\x \phi_\lambda (\x,\y)$$
where $\phi_\lambda (\x,\y)$ is nothing but the magnetic phase defined in \eqref{hc9}. 

An identity which plays a fundamental role in the gauge covariant magnetic perturbation theory is
\begin{align}\label{hc22}
({\bf P}_\x -\mathcal{A}_\lambda (\x))e^{\iu \phi_\lambda(\x,\y)}=e^{\iu \phi_\lambda(\x,\y)}({\bf P}_\x -\mathcal{A}_\lambda (\x,\y)), \quad {\bf P}_\x := - \iu \nabla_\x \,.
\end{align}
For the constant magnetic field $b$ we introduce the linear magnetic potential $b\mathbf{A}(\x)=\frac{b}{2}(-x_2,x_1)$ with magnetic phase $b\,\phi(\x,\x')$ (see \eqref{Peierlsphase}), and we have the identity
\begin{align}\label{hc23}
({\bf P}_\x -b\mathbf{A}(\x))e^{\iu b\phi(\x,\y)}=e^{\iu b\phi(\x,\y)}({\bf P}_\x -b\mathbf{A}(\x-\y)).
\end{align}

Let us recall a general result about  the resolvent of any magnetic Schr\"odinger operator $H=\frac{1}{2}({\bf P}-\mathbf{a})^2+V$ with a bounded magnetic field (the magnetic potential may grow) and a bounded electric potential, not necessarily periodic.  Let $K\subset \rho(H)$ be a compact subset of the resolvent set of $H$. Then there exist two constants $\alpha,C>0$ such that for every $z\in K$ the resolvent $(H-z {\bf 1})^{-1}$ has an integral kernel $(H-z {\bf 1})^{-1}(\x;\x')$ which is continuous outside the diagonal $\x=\x'$ and moreover \cite{Simon,BroderixHundLes00}
\begin{equation}
\label{IntKernelResolvent}
\sup_{z \in K}\left| (H-z {\bf 1})^{-1}(\x;\x') \right| \leq C \ln\left (2+\|\x-\x'\|^{-1}\right )e^{-\alpha\|\x-\x'\|} \, , \qquad \forall \x\neq \x' \in \R^2\,.
\end{equation}
This shows that the resolvent's kernel behaves like the one of the free Laplace operator in two dimensions. The constants $\alpha$ and $C$ can be chosen to be independent of the magnitude of the magnetic field due to the diamagnetic inequality. The exponential decay is a consequence of Combes-Thomas estimates \cite{CT73, CorneanNenciu09}. 

In the case of a purely magnetic Landau operator $H\sub{Landau}:=\frac{1}{2}({\bf P} - b\mathbf{A})^2$ its resolvent admits an explicit kernel of the type 
$$(H\sub{Landau}+{\bf 1})^{-1}(\x;\x')=e^{\iu b \phi(\x,\x')}F(\|\x-\x'\|)$$
where $F$ decays exponentially at infinity (it is in fact a Gaussian if $b\neq 0$) and has a local logarithmic singularity, see \cite{CorneanNenciu98}. Also, using \eqref{hc23} one can show that there exist $\alpha,C>0$ such that  
\begin{equation}
\label{hc24}
\left| ({\bf P}_\x -b\mathbf{A}(\x)) \, (H\sub{Landau}+{\bf 1})^{-1}(\x;\x') \right| \leq C  \|\x-\x'\|^{-1}e^{-\alpha\|\x-\x'\|} \, , \qquad \forall \x\neq \x' \in \R^2\,.
\end{equation}

We are interested in the integral kernel of the resolvents of 
$$H_\lambda = \frac{1}{2} ({\bf P} - b\mathbf{A}-\mathcal{A}_\lambda)^2+V,\quad H_0=\frac{1}{2}({\bf P} - b\mathbf{A})^2+V.$$
Without loss of generality we may assume that the spectrum of $H_0$ is non-negative. The second resolvent identity
$$(H_0+{\bf 1})^{-1}=(H\sub{Landau}+{\bf 1})^{-1}-(H\sub{Landau}+{\bf 1})^{-1}V(H_0+{\bf 1})^{-1},$$
together with \eqref{IntKernelResolvent}, \eqref{hc24} and the fact that $V$ is bounded, lead to the existence of $\alpha,C>0$ such that
\begin{equation}
\label{hc25}
\left| ({\bf P}_\x -b\mathbf{A}(\x))(H_0+{\bf 1})^{-1}(\x;\x') \right| \leq C  \|\x-\x'\|^{-1}e^{-\alpha\|\x-\x'\|} \, , \qquad \forall \x\neq \x' \in \R^2\,.
\end{equation}
Now if $K$ is some compact set in $\rho(H_0)$ and $z\in K$, then from the first resolvent identity
$$(H_0-z {\bf 1})^{-1}=(H_0+{\bf 1})^{-1}+(z+1)(H_0+{\bf 1})^{-1}(H_0-z {\bf 1})^{-1}$$
together with \eqref{IntKernelResolvent} and \eqref{hc25} we conclude that there exist $\alpha,C>0$ such that 
\begin{equation}
\label{hc26}
\sup_{z\in K}\left| ({\bf P}_\x -b\mathbf{A}(\x))(H_0-z {\bf 1})^{-1}(\x;\x') \right| \leq C  \|\x-\x'\|^{-1}e^{-\alpha\|\x-\x'\|} \, , \qquad \forall \x\neq \x' \in \R^2\,.
\end{equation}

We are now ready to deal with the magnetic perturbation induced by $\mathcal{A}_\lambda$. If $z\in \rho(H_0)$ we define the operator $S_\lambda(z)$ given by the integral kernel
\begin{equation} \label{eqn:S(z)}
S_\lambda(z)(\x;\x'):=e^{\iu \phi_\lambda (\x,\x')}(H_0-z {\bf 1})^{-1}(\x;\x').
\end{equation}
From \eqref{IntKernelResolvent} we see that $|S_\lambda(z)(\x;\x')|$ is pointwise bounded by a function of $\x-\x'$ which is in $L^1(\R^2)$, thus via Schur's criterion $S_\lambda(z)$ defines a bounded operator. The main observation is that the range of $S_\lambda(z)$ lies in the domain of $H_\lambda -z {\bf 1}$ and using \eqref{hc22} we have 
\begin{equation} \label{eqn:T(z)}
(H_\lambda -z {\bf 1})S_\lambda(z)=:{\Id} + T_\lambda(z)
\end{equation}
where $T_\lambda(z)$ has an integral kernel given by
\begin{multline*}
T_\lambda(z)(\x;\x'):=-2e^{\iu \phi_\lambda (\x,\x')}\mathcal{A}_\lambda(\x,\x')\cdot ({\bf P}_\x -b\mathbf{A}(\x))(H_0-z {\bf 1})^{-1}(\x;\x')\\
+e^{\iu \phi_\lambda (\x,\x')}\left \{|\mathcal{A}_\lambda(\x,\x')|^2-\iu \; {\rm div}_\x \mathcal{A}_\lambda(\x,\x')\right \}(H_0-z {\bf 1})^{-1}(\x;\x').
\end{multline*}
From this formula we see, by using \eqref{hc21'}, \eqref{IntKernelResolvent} and \eqref{hc26}, that $|T_\lambda(z)(\x;\x')|$ is also bounded by an $L^1(\R^2)$-function of $\x-\x'$, namely
\begin{equation} \label{dm2}
|T_\lambda(z)(\x;\x')|\leq C \lambda e^{-\alpha \|\x-\x'\|}.
\end{equation}
The factor $\lambda$ ensures that $\|T_\lambda(z)\|\leq C\lambda < 1$ uniformly in $z\in K$ if $\lambda$ is small enough. Hence we have that $H_\lambda -z {\bf 1}$ is invertible and 
$$ (H_\lambda-z {\bf 1})^{-1}=S_\lambda(z)\{{\Id} +T_\lambda(z)\}^{-1}.$$
Multiplying both sides of \eqref{eqn:T(z)} by $(H_\lambda-z {\bf 1})^{-1}$ then yields a resolvent-like identity:
\begin{align}\label{hc28}
(H_\lambda-z {\bf 1})^{-1}=S_\lambda(z)-(H_\lambda-z {\bf 1})^{-1}T_\lambda(z).
\end{align}

We have just proved that the gaps in the spectrum of $H_0$ are stable. Thus if $\sigma_0$ is an isolated spectral island of $H_0$  and $\mathcal{C}\subset \rho(H_0)$ is a positively oriented simple contour which encircles $\sigma_0$, then $\mathcal{C}$ also belongs to $\rho(H_\lambda)$ if $\lambda$ is small enough and we can define two Riesz projections as
$$\Pi_{\lambda}= \frac{\iu}{2\pi} \oint_{\mathcal{C}} \, \left(H_\lambda-z {\bf 1}\right)^{-1} dz, \quad \Pi_{0}= \frac{\iu}{2\pi} \oint_{\mathcal{C}} \, \left(H_0-z {\bf 1}\right)^{-1} dz.
$$
Using \eqref{IntKernelResolvent} and the identities $\Pi_\lambda=\Pi_\lambda^2$ and $\Pi_0=\Pi_0^2$ one can show that the integral kernels of both projections are no longer singular and, at the same time, they have an exponential localization near the diagonal. Moreover, by applying the Riesz integral to \eqref{hc28}, using the explicit expression \eqref{eqn:S(z)} for $S_\lambda(z)(\x;\x')$, and noting that 
\[ \sup_{z \in \mathcal{C}} | \{(H_\lambda - z {\bf 1})^{-1} T_\lambda(z)\}(\x;\x') | \le C \, \lambda \, e^{-\alpha \|\x-\x'\|} \]
from \eqref{dm2}, we finish the proof of \eqref{hc30}. 

\section{Kernel localization of the Kato-Nagy unitary} \label{app:KatoNagy}
In this appendix we explicitly prove that the integral kernel of the Kato-Nagy unitary appearing in the proof of Corollary \ref{crl:Main} satisfies the localization estimate \eqref{RequirementForU}.  The strategy of the proof, based on results presented in \cite{CorneanNenciu09}, consists in proving that the building blocks of the Kato-Nagy unitary are the product of two operators which have an exponentially localized integral kernel as in \eqref{hc2}.

\begin{lemma}
	\label{LemmaKNUnitary}
	Consider two projections $\F_{b+\epsilon}$ and $\F_b$, which are  spectral projections of the Hamiltonians $H_{b+\epsilon}$ and $H_{b}$ respectively. Assume that  $\F_{b+\epsilon}$ and $\F_b$ have jointly continuous integral kernels satisfying \eqref{hc2} and such that 
	$$
	\left\|\F_{b+\epsilon}- \F_{b} \right\| = c < 1 \; .
	$$
	Then, there exists an unitary operator $U$ such that $\F_{b}=U\F_{b+\epsilon}U^*$ and the integral kernel of $U$ satisfies \eqref{RequirementForU}.
\end{lemma}
\begin{proof}
	Define the operator $D:= \F_{b+\epsilon}-\F_b$. Since the kernel of both operators is exponentially localized, see \eqref{hc2}, one can find two constants $\alpha$ and $C$ that do not depend on $\epsilon$ such that
	\begin{equation}
	\label{AuxC3}
	| D(\x;\x') | \leq C e^{-\alpha\|\x-\x'\|} \, .
	\end{equation}
	Then, for all $0\leq \delta < \delta_0 < \alpha$, it holds
	\begin{equation}
	\label{AuxC}
	\left|D(\x;\x') e^{\delta \|\x -\x'\|} - D(\x;\x')\right|\leq C \delta e^{-\frac{\alpha}{2} \|\x-\x'\|} \, ,
	\end{equation}
	which is a simple consequence of $|e^{\|\x\|} -1 | \leq \|\x\| e^{\|\x\|}$. Hence, choosing $\delta$ small enough, using the triangle inequality and \eqref{AuxC} together with a Schur-Holmgren estimate, we obtain
	\begin{equation}
	\label{AuxC2}
	\sup_{ \x_0 \in \R^2} \left\| e^{\pm\delta \|\cdot -\x_0\| } D e^{\mp\delta \|\cdot -\x_0\| }\right\| < 1 \, .
	\end{equation}
	Since $c<1$, the unitary $U$ is given by the well known Kato-Nagy unitary, see \cite{Kato66}. From the explicit formula of $U$, we have that
	$$
	U-{\bf 1} = \left(\left({\bf 1}- D^2\right)^{-\frac{1}{2}} - {\bf 1} \right) \left( {\bf 1} + 2 \F_{b+\epsilon} \F_{b} - \F_{b} - \F_{b+\epsilon}\right) + 2 \F_{b+\epsilon} \F_{b} - \F_{b} - \F_{b+\epsilon}\, .
	$$
	The projections $\F_{b+\epsilon}$ and $\F_{b}$ have an exponentially localized integral kernel, see \eqref{hc2}, hence we only have to prove that the operator $\left(\left({\bf 1}- D^2\right)^{-\frac{1}{2}} - {\bf 1} \right) $ has an integral kernel that is exponentially localized, namely that satisfies an estimate analogous to \eqref{hc2}.
	Therefore, consider
	$$
	\begin{aligned}
	\left(\left({\bf 1}- D^2\right)^{-\frac{1}{2}} - {\bf 1} \right) &= D \left(\sum^{\infty}_{k=0} \frac{(2k+1)!!}{(k+1)!2^{(k+1)}} D^{2k} \right) D =: D \left(\sum^{\infty}_{k=0} a_k D^{2k} \right)D\, .
	\end{aligned}
	$$
	From the estimate \eqref{AuxC2} we get that 
	$$
	\sup_{ \x_0 \in \R^2} \left\|e^{-\delta \|\cdot -\x_0\|}\left(\sum^{\infty}_{k=0} a_k D^{2k} \right) e^{\delta \|\cdot -\x_0\|} \right\| \leq C \, \, .
	$$
	Moreover, from the estimate $\left| e^{-\delta \|\x -\x_0\|} D(\x;\x') e^{\delta \|\x' -\x_0\|} \right| \leq  C e^{-\frac{\alpha}{2} \|\x-\x'\|}$ and the Cauchy-Schwarz inequality, one deduces that 
	$$
	\sup_{\x_0 \in \R^2} \left\| e^{-\delta \|\cdot -\x_0\|} D e^{\delta \|\cdot -\x_0\|} \right\|_{\mathcal{B}(L^2,L^\infty)} \leq C \, ,
	$$
	where $\mathcal{B}(L^2,L^\infty)$ is the space of bounded operators from $L^2(\R^2)$ to $L^\infty(\R^2)$. The previous two estimates imply that
	\begin{equation}
	\label{AuxAppC2}
	\sup_{ \x_0 \in \R^2} \left\|e^{-\delta \|\cdot -\x_0\|}D\left(\sum^{\infty}_{k=0} a_k D^{2k} \right) e^{ \delta \|\cdot -\x_0\|} \right\|_{\mathcal{B}(L^2,L^\infty)} \leq C \, ,
	\end{equation}
	hence $e^{-\delta \|\cdot -\x_0\|}\left(\sum^{\infty}_{k=0} a_k D^{2k+1} \right)  e^{ \delta \|\cdot -\x_0\|} $ is a Carleman operator and in particular has an integral kernel. Furthermore, notice that $\F_{b+\epsilon}$ and $\F_{b}$ map the Hilbert space into the domain of the operator $H_{b+\epsilon}$ and $H_b$ respectively. A standard argument involving the Sobolev embedding shows that the domain of the Hamiltonians $H_{b+\epsilon}$ and $H_b$ can be embedded in the space of continuous functions. This implies that also the operator $e^{-\delta \|\cdot -\x_0\|} \left(\sum^{\infty}_{k=0} a_k D^{2k+1} \right) e^{ \delta \|\cdot -\x_0\|}$ maps the Hilbert space into the space of continuous functions. Hence, mimicking the strategy in \cite[Proposition 3.1]{CorneanNenciu09}, for every $\psi \in C^{\infty}_0(\R^2)$, we define the linear functional
	$$
	C^{\infty}_0 \ni \psi \mapsto \int_{\R^2} d\x  \left(\sum^{\infty}_{k=0} a_k D^{2k+1} \right)(\x_0;\x) e^{ \delta \|\x -\x_0\|} \psi(\x) \, . 
	$$
	Using \eqref{AuxAppC2}, we can extend the previous linear functional to the whole Hilbert space $L^2(\R^2)$. Therefore, Riesz's representation theorem implies that
	\begin{equation}
	\label{AuxAppC}
	\begin{aligned}
	&\sup_{ \x_0 \in \R^2} \left\|e^{\delta \|\cdot -\x_0\|}\left(\sum^{\infty}_{k=0} a_k D^{2k+1} \right)(\x_0;\cdot )  \right\|_{L^2(\R^2)}  \\
	&=  \sup_{ \x_0 \in \R^2} \left\|e^{\delta \|\cdot -\x_0\|}\left(\sum^{\infty}_{k=0} a_k D^{2k+1} \right)(\cdot;\x_0 )  \right\|_{L^2(\R^2)} \leq C \, .
	\end{aligned}
	\end{equation}
	Using \eqref{AuxAppC}, together with \eqref{AuxC3}, the Cauchy-Schwarz and the triangle inequality, we eventually obtain that 
	$$
	\begin{aligned}
	&\sup_{\x,\x' \in \R^2} e^{\delta \|\x-\x'\|} \left| \left(\left({\bf 1}- D^2\right)^{-\frac{1}{2}} - {\bf 1} \right) (\x;\x') \right| \\
	& \leq \sup_{\x,\x' \in \R^2} \, \int_{\R^2} d\tx \; e^{\delta \|\x-\tx\|} \left|D(\x;\tx)\right| e^{\delta \|\tx-\x'\|} \left|\left(\sum^{\infty}_{k=0} a_k D^{2k+1} \right)(\tx; \x' )  \right| \\
	&\leq \sup_{ \x \in \R^2} \left\|e^{\delta \|\cdot-\x\|} \left|D(\x;\cdot)\right| \right\|_{L^2(\R^2)}  \sup_{ \x' \in \R^2} \left\|e^{\delta \|\cdot -\x'\|}\left| \left(\sum^{\infty}_{k=0} a_k D^{2k+1} \right)(\cdot; \x' ) \right| \right\|_{L^2(\R^2)}  \leq C \, .
	\end{aligned}
	$$
	Thus the integral kernel of the unitary operator $U$ satisfies \eqref{RequirementForU}.
\end{proof}

\bigskip

\noindent\textsc{Acknowledgments.}
{The Authors would like to thank J.~Bellissard, G.~Nenciu, G.~Panati and S.~Teufel for inspiring discussions.  H.~C. and M.~M. gratefully acknowledge the support of the Simons-CRM Scholar-in-residence program during the preparation of this work, and the financial support from Grant 8021-00084B of the Danish Council for Independent Research $|$ Natural Sciences. The work of D.\ M.\ has been supported by the European Research Council (ERC) under the European Union’s Horizon 2020 research and innovation programme (ERC CoG UniCoSM, grant agreement n.724939). }


\vfill

{\footnotesize

\begin{tabular}{rl}
(H.\ D.\ Cornean) & \textsc{Department of Mathematical Sciences, Aalborg University} \\
 &  Skjernvej 4A, 9220 Aalborg, Denmark \\
 &  \textsl{E-mail address}: \href{mailto:cornean@math.aau.dk}{\texttt{cornean@math.aau.dk}} \\
 \\
 (D.\ Monaco) & \textsc{Dipartimento di Matematica, ``La Sapienza'' Universit\`{a} di Roma} \\
&   Piazzale Aldo Moro 2, 00185 Roma, Italy \\
&  \textsl{E-mail address}: \href{mailto:monaco@mat.uniroma1.it}{\texttt{monaco@mat.uniroma1.it}} \\
 \\
(M.\ Moscolari) & \textsc{Department of Mathematical Sciences, Aalborg University} \\
 &  Skjernvej 4A, 9220 Aalborg, Denmark \\
 &  \textsl{E-mail address}: \href{mailto:massimomoscolari@math.aau.dk}{\texttt{massimomoscolari@math.aau.dk}} \\
\end{tabular}

}
\end{document}